\documentclass[12pt,reqno,fleqn]{amsart}
\usepackage{latexsym}
\usepackage{amssymb}
\usepackage{amsxtra}
\usepackage{amsmath}
\usepackage{amsfonts}
\usepackage[dvips]{graphicx}
\usepackage{pstricks,pst-node}
\usepackage[metapost]{mfpic}
\usepackage{mathrsfs}
\usepackage{enumerate}
\newtheorem{theorem}{Theorem}[section]
\newtheorem{lemma}[theorem]{Lemma}

\makeatother       

\makeatletter      
\@addtoreset{equation}{section}
\makeatother       
\begin{document}
\title[]{The compatibility of additional symmetry and  gauge transformations  for the constrained discrete Kadomtsev-Petviashvili hierarchy}
\author{Maohua Li$^{1,3}$,Jipeng Cheng$^{2}$,Jingsong He$^{1,*}$}
\dedicatory {
1. Department of Mathematics, Ningbo University, Ningbo, 315211 Zhejiang, China\\
2. Department of Mathematics, China University of Mining and Technology, Xuzhou, 221116 Jiangsu, China\\
3. Institute for Theoretical Physics, KU Leuven, 3001 Leuven, Belgium\\
limaohua@nbu.edu.cn\\
chengjp@cumt.edu.cn\\
hejingsong@nbu.edu.cn
}
\thanks{$^*$ Corresponding author.}
\begin{abstract}
In this paper, the compatibility between the  gauge
transformations and the additional symmetry of the constrained discrete Kadomtsev-Petviashvili
hierarchy is given, which preserving the form of the additional symmetry of the cdKP hierarchy, up to shifting of
  the corresponding additional flows by ordinary time flows.  
\end{abstract}
\maketitle
\section{Introduction}

The discrete Kadomtsev-Petviashvili (dKP) hierarchy
\cite{Kupershimidt,Kanaga1997,Iliev,LiuS2,zhangdj2}  is  an attractive
research object in the field of the discrete integrable systems.
The dKP hierarchy is defined by means of the difference derivative
$\Delta$ instead of the usual derivative $\partial$ with respect
of $x$ in a classical system \cite{dickey2003,dkjm}, and the continuous
spatial variable is replaced by a discrete variable $n$.
By using a non-uniform shift of space variable, the $\tau$-function of KP hierarchy implies a special kind of $\tau$-function  for the
dKP hierarchy \cite{Iliev}.
With the  symmetry constraint or symmetry reduction technique, which was used in the  continuous KP hierarchy \cite{chengli,kss1991,Cheng92}, the constrained discrete KP (cdKP) hierarchy is truncated dKP hierarchy by adding a constrained operator form (see (\ref{cdkplax})) on the Lax operator $L$ of the dKP hierarchy \cite{lmh20131}.  And  the discrete nonlinear Schr\"odinger equation and other equations can be derived from it.

The gauge transformation is one kind of powerful method to construct the
solutions of the integrable systems for  the continuous KP
hierarchy \cite{chau_cmp1992,oevel1993,oevelRMP93,nimmo,chau1997,hlc2002}, the dKP hierarchy \cite{lmh20131,oevel1996,LiuS} and the cdKP hierarchy \cite{lmh20131}, which in fact reflects the
intrinsic integrability of the KP hierarchy and dKP hierarchy. Chau {\it et al}
\cite{chau_cmp1992} introduce two kinds of elementary gauge
transformation operators: the differential type $T_D$ and the
integral type $T_I$. By now, the gauge transformations of many
integrable hierarchies related to KP hierarchy have been derived,
for example, the constrained KP (cKP)
hierarchy\cite{chengli,kss1991,Cheng92,oevel1993,chau1997,willox,he2003}, the
constrained BKP and CKP hierarchy\cite{nimmo,he2007} (cBKP and cCKP), the
dKP hierarchy\cite{LiuS, oevel1996}, the cdKP hierarchy \cite{lmh20131}, the $q$-KP
hierarchy\cite{tu1998,he2006} and so on. The additional symmetry
\cite{fokas1981,oevel,bfuch1983,Chen1983,OS86,ASM95,D95,Tu07,he2007b,cheng2011,tu2011,tian2011,2li2012}
is a kind of symmetry depending explicitly on the space and time
variables, involved in so-called string equation and the generalized
Virasoro constraints in matrix models of the 2d quantum gravity (see
\cite{dickey2003,vM94} and references therein).
Regarding the possible application of the additional symmetry  flows of the KP hierarchy in physics, it  is natural to ask whether these flows are
compatible with  the gauge transformation. It is a highly non-trivial question because the gauge transformation is only defined to be consistent with
ordinary KP flows. For example, Ref.\cite{aratyn97} has shown  the compatibility between the the  differential type of  gauge transformation and
the additional symmetry flow of cKP hierarchy separately, up to a shift of ordinary flow of cKP hierarchy.  In order to construct the additional
symmetry flows of the cKP hierarchy  from the corresponding flows of the KP hierarchy, it is necessary to do a remarkable amendment \cite{aratyn97} in its definition.  So it is an interesting problem to show the
compatibility between the gauge transformations and the additional
symmetry of the cdKP hierarchy.  The additional symmetry flows for the cdKP hierarchy are  constructed in
\cite{lmh20134,clh2013} through  a subtle  modification of
the standard additional symmetry flows by adding a complicated term,  which form a Virasoro type algebraic structure \cite{lmh20134}. And the action of the  Virasoro symmetry on the tau function of the cdKP hierarchy is also derived \cite{clh2013}.

 In this paper, it is showed that the additional symmetry
  flows for the cdKP hierarchy commute with the integral type and difference type gauge
  transformations preserving the form of the additional symmetry of the cdKP hierarchy, up to shifting of
  the corresponding additional flows by ordinary time flows,
  which reflects  the compatibility between the two types  of the gauge transformations and the additional symmetries of the cdKP hierarchy.

This paper is organized as the follows. Some backgrounds on the dKP hierarchy are recalled in Section 2. Then the two types
gauge transformation operators  of the cdKP hierarchy are reviewed in Section 3. And the additional symmetry for the cdKP
hierarchy are reviewed in  Section 4. In
Section 5, it is derived that  the additional symmetry commute with the gauge
transformations preserving the form of the additional symmetry of  the cdKP hierarchy.

\section{Background on the dKP Hierarchy}
Let us briefly recall some basic facts
about the dKP hierarchy according to reference \cite{Iliev}.
Firstly
a space $F$, namely
\begin{equation}
F=\left\{ f(n)=f(n,t_1,t_2,\cdots,t_j,\cdots);  n\in\mathbb{Z}, t_i\in\mathbb{R}
\right\}
\end{equation}
is defined for the space of the dKP hierarchy.
$\Lambda$ and $\Delta$ are denote for the shift operator and the
difference operator, respectively. Their actions on function $f(n)$ are
defined as
\begin{equation}\Lambda f(n)=f(n+1)
\end{equation}
and
\begin{equation}\Delta f(n)=f(n+1)-f(n)=(\Lambda -I)f(n)
\end{equation} respectively,
where $I$ is the identity operator.

For any $j\in\mathbb{Z},$ the Leibniz  rule of $\Delta$ operation is,

\begin{equation}\Delta^jf=\sum^{\infty}_{i=0}\binom{j}{i}(\Delta^i
f)(n+j-i)\Delta^{j-i},\hspace{.3cm}
\binom{j}{i}=\frac{j(j-1)\cdots(j-i+1)}{i!}.\label{81}
\end{equation}
So an associative ring $F(\Delta)$ of formal pseudo
difference operators (PDO) is obtained,  namely
$F(\Delta)=\left\{R=\sum_{j=-\infty}^d f_j(n)\Delta^j,
f_j(n)\in R, n\in\mathbb{Z}\right\}
$. The
adjoint operator to the $\Delta$ operator is given by
$\Delta^*$,
\begin{equation}
\Delta^* f(n)=(\Lambda^{-1}-I)f(n)=f(n-1)-f(n),
\end{equation}
where $\Lambda^{-1} f(n)=f(n-1)$, and the corresponding $j$-times
operation is
\begin{equation}
\Delta^{*j}
f=\sum^{\infty}_{i=0}\binom{j}{i}(\Delta^{*i}f)(n+i-j)\Delta^{*j-i}.
\end{equation}
Then the adjoint ring $F(\Delta^*)$ to the
$F(\Delta)$ is obtained, and the formal adjoint to $R\in F(\Delta)$ is defined
by $R^*\in F(\Delta^*)$ as $R^*=\sum_{j=-\infty}^d
\Delta^{*j} f_j(n)$. The $"*"$  stands for the conjugate operation which  satisfies the rules as
$(F G)^*=G^* F^*$, $\Delta^*=-\Delta$, $f^*=f$ for two operators $F$ and $G$ and $f(n)^*=f(n)$ for
a function $f(n)$. Here for any (pseudo-) difference operator $A$ and a function $f$,
the symbol $A(f)$ will indicate the action of $A$ on $f$, whereas
the symbol $Af$ (or $A\cdot f$) will denote just operator product of $A$ and $f$.

         The dKP hierarchy \cite{Kanaga1997,Iliev} is a family of evolution equations depending on
infinitely many variables $t=(t_1,t_2,\cdots)$
\begin{equation} \label{floweq}
\frac{\partial L}{\partial t_k}=[B_k, L],\ \ \ B_k:=(L^k)_+,
\end{equation}
where $L$ is a general first-order PDO
\begin{equation} \label{laxoperatordkp}
L(n)=\Delta + \sum_{j=1}^{\infty} f_j(n)\Delta^{-j}.
\end{equation}
 $B_m=(L^m)_+=\sum^m_{j=0}a_j(n)\Delta^j$, i. e.  $(L^m)_+$ is the
non-negative projection of $L^m$, and $(L^m)_-=L^m-(L^m)_+$ is the
negative projection of $L^m$. The Lax operator in
eq.(\ref{laxoperatordkp}) can be generated by a dressing operator
\begin{equation}
W(n;t)=1+\sum^\infty_{j=1}w_j(n;t)\Delta^{-j}.
\end{equation}
through
\begin{equation}
L=W  \Delta  W^{-1}.
\end{equation}
Further the flow equation
(\ref{floweq}) is equivalent to the so-called Sato equation,
\begin{equation}\label{tkaction}
\frac{\partial W}{\partial t_k}=-(L^{k})_- W.
\end{equation}

If the functions $q(t)$ and $r(t)$ satisfy
\begin{equation}\label{eigenfunction}
\frac{\partial q}{\partial t_k}=B_k(q),\quad\quad \frac{\partial r}{\partial
t_k}=-B_k^*(r),
\end{equation}
then we call them the eigenfunction and the adjoint eigenfunction
respectively.

The cdKP hierarchy \cite{lmh20134} is defined by restricting the Lax
operator of the  dKP hierarchy
\begin{equation} \label{cdkplax}
\frac{\partial L}{\partial t_k}=[B_k, L],\ \ \ B_k:=(L^k)_+,
\end{equation}
 with the following
$l$-constrained form:
\begin{equation}\label{cdkplax}
    L^l=L^l_++\sum_{i=1}^mq_i\Delta^{-1}r_i=\Delta^l+\sum_{j=0}^{k-2}v_j\Delta^j+\sum_{i=1}^mq_i\Delta^{-1}r_i,
\end{equation}
where $q_i$ and $r_i$ are the eigenfunction and adjoint
eigenfunction respectively.

\section{The two types gauge transformations of the cdKP hierarchy}

Let $L$ be the original Lax operator of the cdKP hierarchy (\ref{cdkplax}),
and $T$ be a pseduo-difference operator. If the transformation
\begin{equation}\label{gaugegeneral}
    \widetilde{L}=T L T^{-1}
\end{equation}
such that
\begin{equation}\label{tranlaxeq}
 \frac{\partial \widetilde{L}}{\partial t_k}=[\widetilde{B}_k,\widetilde{L}],\quad \widetilde{B}_k=(\widetilde{L}^k)_+,\ k=1,2,3,\cdots
\end{equation}
still holds for transformed Lax operator $\widetilde{L}$, then  $T$ is
called the gauge transformation operator of the cdKP hierarchy.

Similar to the KP hierarchy \cite{chau_cmp1992}, there  are two types of gauge transformation
operators of the dKP hierarchy  as \cite{oevel1996,LiuS}
\begin{eqnarray}
{\rm Type\ I:}&& T_D(q)=\Lambda(q)\Delta q^{-1},\\
{\rm Type\ II:}&& T_I(r)=\Lambda^{-1}(r^{-1})\Delta^{-1} r,
\end{eqnarray}
where $q$ and $r$ are the eigenfunction and adjoint
eigenfunction respectively. The type I transformation is called the
difference type, while the type II is called the integral type.

Here we  review
some results about the integral type and the difference type gauge transformations of the cdKP hierarchy \cite{lmh20131}.  Under the integral type
gauge transformation $T_I(r)$, the transformed Lax operator will
be:
\begin{eqnarray}
  \widetilde{L}&=&T_I(r) L T_I(r)^{-1}=\widetilde{L}_++\widetilde{L}_-,\label{ckplax1}\\
  \widetilde{L}_+&=&\Lambda^{-1}(L_+)-\Lambda^{-1}(r^{-1})\Delta^{-1}(\Delta^*(r^{-1}L_+^{*}r)_{\geq 1})^*(r),\label{ckplax1add}\\
  \widetilde{L}_-&=&\widetilde{q}_0\Delta^{-1}\widetilde{r}_0+\sum_{i=1}^m\widetilde{q}_i\Delta^{-1}\widetilde{r}_i,\label{ckplax1minus}\\
  \widetilde{q}_0&=&\Lambda^{-1}(r^{-1}),\quad \widetilde{r}_0=T_I(r)^{*-1}L^{(0)*}(r),\label{qr01}\\
  \widetilde{q}_i&=&T_I(r)(q_i),\quad
  \widetilde{r}_i=T_I(r)^{*-1}(r_i)\label{qri1}.
\end{eqnarray}
In order to preserve the form (\ref{cdkplax}) of the Lax operator
$L$, $r$ is required to coincide with one of the original
adjoint eigenfunctions of $L$, e.g. $r=r_1$, since
$\widetilde{r}_1=0$ in this case.

Under the gauge transformation of $T_D(q)$, the transformed Lax operator reads as
\begin{eqnarray}
\widetilde{L}&=& \widetilde{L}_++\widetilde{L}_-,\\
\widetilde{L}_+&=&\Lambda(L_+)+\Lambda(q)\Delta(q^{-1}L_+q)_{\geq1} \Delta^{-1}\Lambda(q^{-1}),\\
\widetilde{L}_-&=&\widetilde{q}_0\Delta^{-1}\widetilde{r}_0+\sum_{i=1}^{m}\widetilde{q}_i\Delta^{-1}\widetilde{r}_i,\\
\widetilde{q}_0&=&T_D(q)(L)(q),\ \ \widetilde{r}_0=\Lambda(q^{-1}),\\
\widetilde{q}_i&=&T_D(q)q_i,\ \ \widetilde{r}_i=(T_D^{-1})^{*}(q)(r_i).
\end{eqnarray}
For the difference type gauge transformation $T_D(q)$, in order to preserve the form (\ref{cdkplax}) of the Lax operator
$L$, $q$ is required to coincide with one of the original
adjoint eigenfunctions of $L$, e.g. $q=q_1$, since
$\widetilde{q}_1=0$ in this case.

In order to calculate the transformed formula of the part as $f\Delta^{-1}g$ in the Lax operator under the integral type gauge transformation,
the following lemma is necessary.
\begin{lemma}\label{gaugelemma}
\begin{eqnarray}
T_I(r_a)\cdot M\Delta^{-1}r_a\cdot T_I^{-1}(r_a)&=&\Lambda^{-1}{(r_a^{-1})}\Delta^{-1}\left\{T_I^{*-1}(r_a)(M\Delta^{-1}r_a)^*(r_a)\right\},\label{tiMr}\\
T_I(r_a)\cdot q_a\Delta^{-1}N \cdot T_I(r_a)^{-1}&=&\Lambda^{-1}{(r_a^{-1})}r_a^{-1}\Delta^{-1}\left\{T_I^{*-1}(r_a)(q_a\Delta^{-1}N)^*(r_a)\right\}\nonumber\\
&&+\widetilde{L}(\widetilde{q}_a)\Delta^{-1}\widetilde{N},
\label{tiqN}\\
T_I(r_a)\cdot M\Delta^{-1}N \cdot  T_I(r_a)^{-1}&=&\Lambda^{-1}{(r_a^{-1})}\Delta^{-1}\left\{T_I^{*-1}(r_a)(M\Delta^{-1}N)^*(r_a)\right\}\nonumber\\
&&+\widetilde{M}\Delta^{-1}\widetilde{N},\label{tiMN}\\
\widetilde{L}^{k+1}(\widetilde{q}_a)&=&T_I(r_a)L^{k}(q_a),\ \ k=0,1,2,....,\label{lqa}\\
(\widetilde{L}^*)^{k-1}(\widetilde{r}_a)&=&T_I(r_a)^{*-1}(L^*)^k(r_a),\
\ k=1,2,3....\label{lra}
\end{eqnarray}
where $r_a$ is one of the adjoint eigenfunctions of the cdKP
hierarchy (\ref{cdkplax}), $M$ and $N$ are two functions of $t$, and
\begin{eqnarray}
\widetilde{L}=T_I(r_a) L T_I(r_a)^{-1},\ \ \widetilde{q}_a=\Lambda^{-1}(1/r_a),\ \
\widetilde{r}_a=T_I(r_a)^{*-1} L^{*}(r_a),\nonumber\\
\widetilde{M}=T_I(r_a)(M),\ \
\widetilde{N}=T_I(r_a)^{*-1}(N). \label{symbols}
\end{eqnarray}
\end{lemma}

\begin{proof}
Firstly, according to $\Delta^{-1}f\Delta^{-1}=(\Delta^{-1}f)\Delta^{-1}-\Delta^{-1}\Lambda(\Delta^{-1}f)$ and
$\Delta f-\Lambda(f)\Delta=\Delta(f)$,
\begin{eqnarray*}
&&T_I(r_a)\cdot M\Delta^{-1}N\cdot T_I(r_a)^{-1}=\Lambda^{-1}(r_a^{-1})\Delta^{-1}r_a\cdot M\Delta^{-1}N\cdot r_a^{-1}\Delta\Lambda^{-1}(r_a)\\
&=&\Lambda^{-1}(r_a^{-1})\bigg(\Delta^{-1}(r_aM)\Delta^{-1}-\Delta^{-1}\Lambda(\Delta^{-1}(r_aM))\bigg)Nr_a^{-1}\Delta
\Lambda^{-1}(r_a)\\
&=&\Lambda^{-1}(r_a^{-1})\Delta^{-1}(r_aM)\Delta^{-1}\bigg(\Delta\Lambda^{-1}(Nr_a^{-1})-\Delta\Lambda^{-1}(Nr_a^{-1})\bigg)
\Lambda^{-1}(r_a)\\
&&-\Lambda^{-1}(r_a^{-1})\Delta^{-1}\bigg(\Delta\Lambda^{-1}(\Lambda(\Delta^{-1}(r_aM))Nr_a^{-1})- \Delta(\Delta^{-1}(r_aM)\Lambda^{-1}(Nr_a^{-1}))\bigg)\Lambda^{-1}(r_a)\\
&=&\Lambda^{-1}(r_a^{-1})\Delta^{-1}(r_aM)
\Lambda^{-1}(N)-\Lambda^{-1}(r_a^{-1})\Delta^{-1}(r_aM)\Delta^{-1}\bigg(\Delta\Lambda^{-1}(Nr_a^{-1})\bigg)
\Lambda^{-1}(r_a)\\
&&-\Lambda^{-1}(r_a^{-1})\Delta^{-1}(r_aM)\Lambda^{-1}(N)+\Lambda^{-1}(r_a^{-1})\Delta^{-1}
\Delta\bigg(\Delta^{-1}(r_aM)\Lambda^{-1}(Nr_a^{-1})\bigg)\Lambda^{-1}(r_a)\\
&=&T_I(r_a)(M)\Delta^{-1}T_I(r_a)^{*-1}(N)+\Lambda^{-1}(r_a^{-1})\Delta^{-1}\left\{T_I(r_a)^{*-1}(M\Delta^{-1}N)^*(r_a)\right\}\\
&=&\widetilde{M}\Delta^{-1}\widetilde{N}+\Lambda^{-1}(r_a^{-1})\Delta^{-1}\left\{T_I(r_a)^{*-1}(M\Delta^{-1}N)^*(r_a)\right\}.
\end{eqnarray*}
So (\ref{tiMN}) is be proved.
(\ref{tiMr}) can be derived from (\ref{tiMN}) for $N=r_a$, since
$T_I(r_a)^{*-1}(r_a)=0$.

Then for (\ref{lqa}),
\begin{eqnarray*}
\widetilde{L}^{k+1}(\widetilde{q}_a)&=&T_I(r_a)L^{k+1}T_I(r_a)^{-1}(\Lambda^{-1}(r_a^{-1}))\\
&=& T_I(r_a)L^{k}(L_++\sum_{i=0}^k q_i\Delta^{-1}r_i)r_a^{-1}\Delta
\Lambda^{-1}(r_a)\Lambda^{-1}(r_a^{-1})= T_I(r_a)L^{k}(q_a).
\end{eqnarray*}
Here we let $q_i\Delta^{-1}(0)=0$ for $i\neq a$, and
$q_a\Delta^{-1}(0)=q_a$.
And (\ref{tiqN}) can be derived from (\ref{tiMN}) and (\ref{lqa}).
At last,
\begin{eqnarray*}
T_I(r_a)^{*-1}(L^*)^k(r_a)=T_I(r_a)^{*-1}(L^*)^{k-1}T_I(r_a)^{*}T_I(r_a)^{*-1}L^*(r_a)=(\widetilde{L}^*)^{k-1}(\widetilde{r}_a).
\end{eqnarray*}
\end{proof}

In order to calculate the transformed formula of the part as $f\Delta^{-1}g$ in the Lax operator under the difference  type gauge transformation,
the following lemma is necessary.
\begin{lemma}\label{gaugelemmatd}
\begin{eqnarray}
T_D(q_a)\cdot M\Delta^{-1}r_a \cdot T_D(q_a)^{-1}\mbox{\hspace{-0.3cm}}&=&\mbox{\hspace{-0.3cm}}T_D{(q_a)}(M\Delta^{-1}r_a)(q_a)\Delta^{-1}\Lambda(q_a^{-1})+\widetilde{M}\Delta^{-1}\widetilde{L}^*(\widetilde{r}_a),\label{tdMr}\\
T_D(q_a)\cdot q_a\Delta^{-1}N \cdot T_D(q_a)^{-1}\mbox{\hspace{-0.3cm}}&=&\mbox{\hspace{-0.3cm}}T_D{(q_a)}(q_a\Delta^{-1}N)(q_a)\Delta^{-1}\Lambda(q_a^{-1}),\label{tdqN}\\
T_D(q_a)\cdot M\Delta^{-1}N \cdot T_D(q_a)^{-1}\mbox{\hspace{-0.3cm}}&=&\mbox{\hspace{-0.3cm}}T_D{(q_a)}(M\Delta^{-1}N)(q_a)\Delta^{-1}\Lambda(q_a^{-1})+\widetilde{M}\Delta^{-1}\widetilde{N},\label{tdMN}\\
\widetilde{L}^{k-1}(\widetilde{q}_a)&=&\mbox{\hspace{-0.3cm}}T_D(q_a)L^{k}(q_a),\ \ k=0,1,2,...,\label{tdqa}\\
(\widetilde{L}^*)^{k}(\widetilde{r}_a)&=&\mbox{\hspace{-0.3cm}}T_D^{*}(q_a)^{-1}(L^*)^{k-1}(r_a),\
\ k=1,2,3,...,\label{tdra}
\end{eqnarray}
where $r_a$ is one of the adjoint eigenfunctions of the cdKP
hierarchy (\ref{cdkplax}), $M$ and $N$ are two functions of $t$, and
\begin{eqnarray}
\widetilde{L}=T_D(q_a) LT_D(q_a)^{-1},\ \ \widetilde{q}_a=T_D(q_a)L(q_a),\ \
\widetilde{r}_a=\Lambda(r_a^{-1}),\nonumber\\
\widetilde{M}=T_D(q_a)(M),\ \
\widetilde{N}=T_D(q_a)^{*-1}(N). \label{symbols}
\end{eqnarray}
\end{lemma}
\begin{proof}
Firstly, according to $\Delta^{-1} f\Delta^{-1}=(\Delta^{-1}f)\Delta^{-1}-\Delta^{-1} \Lambda(\Delta^{-1}f)$ and
$\Delta f-\Lambda(f)\Delta=\Delta(f)$,
\begin{eqnarray*}
&&T_D(q_a)\cdot M\Delta^{-1}N\cdot T_D(q_a)^{-1}=\Lambda(q_a)\Delta q_a^{-1}\cdot M\Delta^{-1}N\cdot q_a\Delta^{-1}\Lambda(q_a^{-1})\\
&=&\Lambda(q_a)\Delta(q_a^{-1} M)\bigg(\Delta^{-1}(N q_a)\Delta^{-1}-\Delta^{-1}\Lambda(\Delta^{-1}(q_aN))\bigg)
\Lambda(q_a^{-1})\\
&=&\Lambda(q_a)\bigg(\Delta(q_a^{-1}M\Delta^{-1}(Nq_a))+\Lambda(q_a^{-1}M\Delta^{-1}(Nq_a^{-1}))\Delta \bigg)\Delta^{-1} \Lambda(q_a^{-1})\\
&&-\Lambda(q_a)\bigg(\Delta(q_a^{-1}M)+\Lambda(q_a^{-1}M)\Delta\bigg)\Delta^{-1}
\Lambda(\Delta^{-1}(Nq_a))\Lambda^{-1}(q_a^{-1})\\
&=&\Lambda(q_a)\Delta\bigg(q_a^{-1}M\Delta^{-1}(Nq_a)\bigg)\Delta^{-1}
\Lambda(q_a^{-1})+\Lambda(q_a)\Lambda\bigg(q_a^{-1}M\Delta^{-1}(Nq_a)\bigg)
\Lambda(q_a^{-1})\\
&&-\Lambda(q_a)
\Delta(q_a^{-1}M)\Delta^{-1}\Lambda\bigg(\Delta^{-1}(Nq_a)\bigg)\Lambda(q_a^{-1})-\Lambda(q_a)\Lambda\bigg(q_a^{-1}M\Delta^{-1}(Nq_a)\bigg)
\Lambda(q_a^{-1})\\
&=&\Lambda(q_a)\Delta\bigg(q_a^{-1}M\Delta^{-1}(Nq_a)\bigg)\Delta^{-1}
\Lambda(q_a^{-1})-\Lambda(q_a)
\Delta(q_a^{-1}M)\Delta^{-1}\cdot\Lambda\Delta^{-1}(Nq_a)\Lambda(q_a^{-1})\\
&=&T_D(q_a)(M)\Delta^{-1}T_D(q_a)^{*-1}(N)+T_D(q_a)(M\Delta^{-1}N)(q_a)\Delta^{-1}\Lambda(q_a^{-1})\\
&=&\widetilde{M}\Delta^{-1}\widetilde{N}+T_D(q_a)(M\Delta^{-1}N)(q_a)\Delta^{-1}\Lambda(q_a^{-1}).
\end{eqnarray*}
So (\ref{tdMN}) is be proved.
(\ref{tdqN}) can be derived from (\ref{tdMN}) for $M=q_a$, since
$T_D(q_a)(q_a)=0$.

For (\ref{tdqa}),
\begin{eqnarray*}
\widetilde{L}^{k-1}(\widetilde{q}_a)&=&T_D(q_a)L^{k-1}T_D(q_a)^{-1}T_D(q_a)L(q_a^{-1})
= T_D(q_a)L^{k}(q_a).
\end{eqnarray*}
Then (\ref{tdMr}) can be  derived from (\ref{tdMN}) and (\ref{tdra}).

At last, for (\ref{tdra}),
\begin{eqnarray*}
T_D^{*}(q_a)^{-1}(L^*)^{k+1}(r_a)=T_D^{*}(q_a)^{-1}(L^*)^{k}T_D^{*}(q_a)T_D^{*}(q_a)^{-1}L^*(r_a)=(\widetilde{L}^*)^{k}(\widetilde{r}_a).
\end{eqnarray*}

\end{proof}

 In order  to prove the compatibility between the two types of the gauge transformation and the additional symmetry, the following operator identities are necessary.
 \begin{lemma}\label{fubu}
 Let $q,r$ be suitable function and $A$ be a PDO, then
\begin{eqnarray}
 \left(\Lambda^{-1}{(r^{-1})}\Delta^{-1}r Ar^{-1}\Delta\Lambda^{-1}(r)\right)_-&=&
   \Lambda^{-1}{(r^{-1})}\Delta^{-1}r A_-r^{-1}\Delta\Lambda^{-1}(r)\nonumber\\
   && -\Lambda^{-1}{(r^{-1})}\Delta^{-1}\Lambda^{-1}{(r)}\Delta(\Lambda^{-1}(r^{-1}A^{*}_+r)),\label{Ti}\\
   \left(\Lambda{(q)}\Delta q^{-1} A q\Delta^{-1}\Lambda(q^{-1})\right)_+&=&
   \Lambda{(q)}\Delta q^{-1} A_+q\Delta^{-1}\Lambda(q^{-1})\nonumber\\
   && -\Lambda{(q)}\Delta(q^{-1}A_+ (q))\Delta^{-1}\Lambda(q^{-1}).\label{Td}
\end{eqnarray}
\end{lemma}
\begin{proof}
With $(K q\Delta^{-1}r)_-=K(q)\Delta^{-1}r, (q\Delta^{-1}r K )_-=q\Delta^{-1}K^*(r)$ for pure-difference operator $K$ \cite{lmh20134},
\begin{eqnarray*}
&&\left(\Lambda^{-1}{(r^{-1})}\Delta^{-1}r Ar^{-1}\Delta\Lambda^{-1}(r)\right)_-\\
&=&
\left(\Lambda^{-1}{(r^{-1})}\Delta^{-1}r A_-r^{-1}\Delta\Lambda^{-1}(r)\right)_-
+\left(\Lambda^{-1}{(r^{-1})}\Delta^{-1}r A_+r^{-1}\Delta\Lambda^{-1}(r)\right)_-\\
&=&
  \Lambda^{-1}{(r^{-1})}\Delta^{-1}r A_-r^{-1}\Delta\Lambda^{-1}(r)
   -\Lambda^{-1}{(r^{-1})}\Delta^{-1}\Lambda^{-1}{(r)}\Delta(\Lambda^{-1}(r^{-1}A^{*}_+r)),
\end{eqnarray*}
so the (\ref{Ti}) is proved.
For (\ref{Td}), with $(K q\Delta^{-1}r)_-=K(q)\Delta^{-1}r$,
\begin{eqnarray*}
\left(\Lambda{(q)}\Delta q^{-1} A q\Delta^{-1}\Lambda(q^{-1})\right)_+
&=&\left(\Lambda{(q)}\Delta q^{-1} A_+ q\Delta^{-1}\Lambda(q^{-1})\right)_+\\
&=&  \Lambda{(q)}\Delta q^{-1} A_+q\Delta^{-1}\Lambda(q^{-1})
   -\Lambda{(q)}\Delta(q^{-1}A_+ (q))\Delta^{-1}\Lambda(q^{-1}).
\end{eqnarray*}
\end{proof}
\noindent{\bf Remark}: This lemma is a difference-analogue of the corresponding identities of PDO given by \cite{oevelRMP93,aratyn97}.
\section{Additional symmetries of the cdKP hierarchy}
Define
\begin{equation}\label{xk1}
    X_k^{(1)}=\sum_{i=0}^{m}\sum_{j=0}^{k-1}\left(j-\frac{1}{2}(k-1)\right)L^{k-1-j}(q_i)\Delta^{-1}(L^*)^j(r_i);\
    \ k\geq1,
\end{equation}
which is the essential to ensure the compatibility of the additional
Virasoro symmetry with the constraints (\ref{cdkplax}) defining the
cdKP hierarchy.
The additional symmetry flows for the cdKP hierarchy, spanning the Virasoro algebra, are given by \cite{lmh20134}:
\begin{equation}\label{addsym}
    \partial_k^*L=[-(M_{\Delta}L^k)_-+X_{k-1}^{(1)},L].
\end{equation}
Here $X_{k-1}^{(1)}$ is (\ref{xk1}) for $m=1$. $M_{\Delta}$ is the Orlov-Schulman operator \cite{OS86} defined in the
dressing the ``bare" $M^{(0)}$ operator:
\begin{equation}\label{0osoperator}
    M^{(0)}=\sum_{l\geq1}lt_l\Delta^{l-k}=X_{(k)}+\sum_{l\geq1}(l+k)t_{l+k}\Delta^l;
    \ \ X_{(k)}=\sum_{l=1}^klt_l\Delta^{l-k}
\end{equation}
that is,
\begin{eqnarray}
M_{\Delta}\mbox{\hspace{-0.3cm}}&=&\mbox{\hspace{-0.3cm}} WM^{(0)}W^{-1}=WX_{(k)}W^{-1}+\sum_{l\geq1}(l+k)t_{l+k}L^l=\sum_{l\geq0}(l+k)t_{l+k}L^l_++(M_{\Delta})_-,\label{osoperator}\\
(M_{\Delta})_-\mbox{\hspace{-0.3cm}}&=&\mbox{\hspace{-0.3cm}}WX_{(k)}W^{-1}-kt_k-\sum_{l\geq1}(l+k)t_{l+k}\frac{\partial
W}{\partial t_l}W^{-1},\label{osminus}
\end{eqnarray}
with (\ref{tkaction}) used in (\ref{osminus}).

Then accordingly, the actions of the additional symmetry flows on
the dressing operators and BA functions are showed that:
\begin{equation}\label{actonwba}
    \partial_k^*W=\left(-(M_{\Delta}L^k)_-+X_{k-1}^{(1)}\right)W;\ \
    \partial_k^*\psi(t,\lambda)=\left(-(M_{\Delta}L^k)_-+X_{k-1}^{(1)}\right)(\psi(t,\lambda)).
\end{equation}
The corresponding actions on the eigenfunctions $q_i$ and the
adjoint eigenfunctions $r_i$ are derived by considering
$(\partial_k^*L)_-$ listed as follows \cite{lmh20134}:
\begin{eqnarray}
\partial_k^*q_i&=&(M_{\Delta}L^k)_+(q_i)+\frac{k}{2}L^{k-1}(q_i)+X_{k-1}^{(1)}(q_i),\label{addsymq}\\
\partial_k^*r_i&=&-(M_{\Delta}L^k)_+^*(r_i)+\frac{k}{2}(L^*)^{k-1}(r_i)-(X_{k-1}^{(1)})^*(r_i).\label{addsymr}
\end{eqnarray}

\section{Additional symmetries versus two types gauge transformations}
In this section, we will restrict to the cdKP hierarchy
((\ref{cdkplax}) for $m=1,l=1$). And thus its Lax operator is given by
\begin{equation}\label{dkplaxm1}
    L=\Delta+q\Delta^{-1}r.
\end{equation}

In order to investigate the changes of the additional symmetries
under the integral type gauge transformation $T_I(r)$, some useful
lemmas are needed.

\begin{lemma}
\begin{eqnarray}\label{transformedx}
\mbox{\hspace{-1cm}} T_I(r)X_{k-1}^{(1)}T_I(r)^{-1}&=&\widetilde{X}_{k-1}^{(1)}
    +\sum_{j=0}^{k-2}\widetilde{L}^{k-j-2}(\widetilde{q})\Delta^{-1}\widetilde{L}^j(\widetilde{r})\nonumber\\
   && +\Lambda{(r^{-1})}\Delta^{-1}\left\{(T^{*}_I(r))^{-1}(X_{k-1}^{(1)}-\frac{k}{2}L^{k-1})^*(r)\right\}.
\end{eqnarray}
\end{lemma}
\begin{proof}
According to Lemma \ref{gaugelemma} and (\ref{xk1}), then
\begin{eqnarray*}
&&T_I(r)X_{k-1}^{(1)}T_I(r)^{-1}\\
&=&-\Lambda^{-1}{(r^{-1})}\Delta^{-1}T_D(r)(X_{k-1}^{(1)})^*(r)
+\sum_{j=1}^{k-2}\left(j-\frac{1}{2}(k-2)\right)\widetilde{L}^{k-j-1}(\widetilde{q})\Delta^{-1}(\widetilde{L}^*)^{j-1}(\widetilde{r})\\
&=& -\Lambda^{-1}{(r^{-1})}\Delta^{-1}T_D(r)(X_{k-1}^{(1)})^*(r)
+\sum_{j=0}^{k-2}\left(j-\frac{1}{2}(k-2)\right)\widetilde{L}^{k-j-2}(\widetilde{q})\Delta^{-1}(\widetilde{L}^*)^{j}(\widetilde{r})\\
&&+\sum_{j=0}^{k-2}\widetilde{L}^{k-j-2}(\widetilde{q})\Delta^{-1}(\widetilde{L}^*)^{j}(\widetilde{r})
-\left(1+k-2-\frac{1}{2}(k-2)\right)\widetilde{q}\Delta^{-1}(\widetilde{L}^*)^{k-2}(\widetilde{r})\\
&=&
-\Lambda^{-1}{(r^{-1})}\Delta^{-1}T_D(r)(X_{k-1}^{(1)})^*(r)+\widetilde{X}_{k-1}^{(1)}+\sum_{j=0}^{k-2}\widetilde{L}^{k-j-2}(\widetilde{q})\Delta^{-1}(\widetilde{L}^*)^j(\widetilde{r})\\
&&-\frac{k}{2}r^{-1}\Delta^{-1}T_I(r)^{*-1}(L^*)^{k-1}(r)\\
&=&\widetilde{X}_{k-1}^{(1)}+\sum_{j=0}^{k-2}\widetilde{L}^{k-j-2}(\widetilde{q})\Delta^{-1}\widetilde{L}^j(\widetilde{r})+r^{-1}\Delta^{-1}\left\{T_I(r)^{*-1}(X_{k-1}^{(1)}-\frac{k}{2}L^{k-1})^*(r)\right\}.
\end{eqnarray*}
\end{proof}
\begin{lemma}
\begin{equation}\label{pakt}
\mbox{\hspace{-1.2cm}}  \partial_k^*T_I(r)\cdot
    T_I(r)^{-1}=\Lambda^{-1}(r^{-1})\Delta^{-1}\left\{T_I(r)^{*-1}\left(-(M_{\Delta}L^k)^*_++\frac{k}{2}(L^*)^{k-1}-(X_{k-1}^{(1)})^*\right)(r)\right\}.
\end{equation}
\end{lemma}
\begin{proof}By (\ref{addsymr}),
\begin{eqnarray*}
 \partial_k^*T_I(r)\cdot T_I(r)^{-1}&=&-T_I(r)\partial_k^*(T_I(r)^{-1})\\
 &=&\Lambda^{-1}(r^{-1})\Delta^{-1}r^{-1}\partial_k^{*}(r)\Delta \Lambda^{-1}(r)-\Lambda^{-1}(r^{-1}\partial_k^*(r))\\
 &=&\Lambda^{-1}(r^{-1})\Delta^{-1}\bigg(\Delta \Lambda^{-1}(\partial_k^{*}(r) r^{-1})-\Delta \Lambda^{-1}(\partial_k^{*}(r) r^{-1})\bigg)\Lambda^{-1}(r)-\Lambda^{-1}(r^{-1}\partial_k^*(r))\\
 &=& -\Lambda^{-1}(r^{-1})\Delta^{-1}(\Delta \Lambda^{-1}(\partial_k^{*}(r) r^{-1}))\Lambda^{-1}(r)\\
 &=& \Lambda^{-1}(r^{-1})\Delta^{-1}T_I(r)^{*-1}\Lambda^{-1}(\partial_k^*r)\\
 &=&\Lambda^{-1}(r^{-1})\Delta^{-1}\left\{T_I(r)^{*-1}\left(-(M_{\Delta}L^k)^*_++\frac{k}{2}(L^*)^{k-1}-(X_{k-1}^{(1)})^*\right)(r)\right\}.
\end{eqnarray*}
\end{proof}

\begin{theorem}\label{addgau}
The additional symmetry flows (\ref{addsym}) for the cdKP
hierarchy ((\ref{cdkplax}) for $m=1,l=1$) commute with the integral type
transformations $T_I(r)$ preserving the form of cdKP hierarchy, up to shifting of
(\ref{addsym}) by ordinary time flows, that is,
\begin{equation}\label{cdkpaddgauge}
    \partial_k^*\widetilde{L}=[-(\widetilde{M}_{\Delta}\widetilde{L}^k)_-+\widetilde{X}_{k-1}^{(1)},\widetilde{L}]-\frac{\partial \widetilde{L}}{\partial
    t_{k-1}}.
\end{equation}
\end{theorem}
\begin{proof}Firstly, by (\ref{addsym}),
\begin{eqnarray}
\partial_k^* \widetilde{L}&=& \partial_k^* T_I(r)\cdot
LT_I(r)^{-1}+T_I(r)\partial_k^*L\cdot
T_I(r)^{-1}-T_I(r)LT_I(r)^{-1}\cdot\partial_k^*T_I(r)\cdot T_I(r)^{-1}\nonumber\\
&=&
\left[T_I(r)\left(-(M_{\Delta}L^k)_-+X_{k-1}^{(1)}\right)T_I(r)^{-1}+\partial_k^*T_I(r)\cdot
T_I(r)^{-1}, \widetilde{L}\right]\label{proof1}
\end{eqnarray}

Then with the help of (\ref{transformedx}), (\ref{pakt}), and the
following useful formula  in Lemma.\ref{fubu},
we have
\begin{eqnarray}
&&T_I(r)\left(-(M_{\Delta}L^k)_-+X_{k-1}^{(1)}\right)T_I(r)^{-1}+\partial_k^*T_I(r)\cdot
T_I(r)^{-1}\nonumber\\
&=&
T_I(r)\left(-M_{\Delta}L^k)_-\right)T^{-1}_I(r)-\Lambda^{-1}(r^{-1})\Delta^{-1}\left\{T^*_I(r)^{-1}\left((M_{\Delta}L^k)^*_+-\frac{k}{2}(L^*)^{k-1}+(X_{k-1}^{(1)})^*\right)(r)\right\}\nonumber\\
&&+\widetilde{X}_{k-1}^{(1)}
+\sum_{j=0}^{k-2}\widetilde{L}^{k-j-2}(\widetilde{q})\Delta^{-1}\widetilde{L}^j(\widetilde{r})+\Lambda^{-1}(r^{-1})\Delta^{-1}\left\{T_I(r)^{*-1}\left((X_{k-1}^{(1)})^*-\frac{k}{2}(L^*)^{k-1}\right)(r)\right\}\nonumber\\
&=&-(\widetilde{M}_{\Delta}\widetilde{L}^k)_-+\widetilde{X}_{k-1}^{(1)}+(\widetilde{L}^{k-1})_-,\label{proof2}
\end{eqnarray}
where the following relation\cite{lmh20134} is used,
\begin{equation}\label{laxminus}
    (\widetilde{L}^{k-1})_-=\sum_{j=0}^{k-2}\widetilde{L}^{k-j-2}(\widetilde{q})\Delta^{-1}\widetilde{L}^j(\widetilde{r}).
\end{equation}
In the above process,
\begin{eqnarray*}
&&-T_I(r)(M_{\Delta}L^k)_-T^{-1}_I(r)-\Lambda^{-1}(r^{-1})\Delta^{-1}T^*_I(r)^{-1}((M_{\Delta}L^k)^*_+(r))\\
&&=-T_I(r)(M_{\Delta}L^k)_-T^{-1}_I(r)+\Lambda^{-1}(r^{-1})\Delta^{-1}(r(M_{\Delta}L^k)_+T_I(r)^{-1})^*)\\
&&=-(\widetilde{M}_{\Delta}\widetilde{L}^k)_-,
\end{eqnarray*} which can be got by means of the identity (\ref{Ti}) of Lemma.\ref{fubu}.

At last, the substituting (\ref{proof2}) into (\ref{proof1}) gives
rise to (\ref{cdkpaddgauge}).
\end{proof}

For the difference type gauge transformation $T_D(q)$, there are some lemma as following.

\begin{lemma}
\begin{eqnarray}\label{transformtd}
\mbox{\hspace{-0.6cm}}  T_D(q)X_{k-1}^{(1)} T^{-1}_D(q)\mbox{\hspace{-0.3cm}}&=&\mbox{\hspace{-0.3cm}}\widetilde{X}_{k-1}^{(1)}-(\widetilde{L}^{k-1})_-
+\left(T_D(q)(X_{k-1}^{(1)}+\frac{k}{2}L^{k-1}(q_i))\right)q\Delta^{-1} \Lambda(q^{-1}).
\end{eqnarray}
\end{lemma}
\begin{proof}
According to Lemma \ref{gaugelemmatd} and (\ref{xk1}), then
\begin{eqnarray*}
&&\mbox{\hspace{-0.3cm}}T_D(q)X_{k-1}^{(1)} T_D(q)^{-1}\\
&=&\Lambda(q) \Delta q^{-1}\sum_{j=0}^{k-2}\left(j-\frac{k-2}{2}\right)L^{k-j-2}(q_i)\Delta^{-1}(L^*)^j(r_i)q\Delta^{-1}\Lambda(q^{-1})\\
 &\stackrel{(\ref{tdMN})}{==}&\sum_{j=0}^{k-2}\left(j-\frac{k-2}{2}\right) {T_D(q)(L^{k-2}(q_i)\Delta^{-1} r_i)(q)\Delta^{-1} \Lambda(q^{-1})
+\widetilde{L}^{k-j-3}(\widetilde{q}_i)\Delta^{-1}(\widetilde{L}^*)^{j+1}(\widetilde{r}_i)}\\
&=&\mbox{\hspace{-0.6cm}}\sum_{j=0}^{k-2}\left(j-\frac{k-2}{2}\right) T_D(q)(L^{k-2}(q_i)\Delta^{-1} r_i)(q)\Delta^{-1} \Lambda(q^{-1})
+\sum_{j=1}^{k-1}\left(j-\frac{k}{2}\right)\widetilde{L}^{k-j-2}(\widetilde{q}_i)\Delta^{-1}(\widetilde{L}^*)^{j}(\widetilde{r}_i)\\
&=&T_D(q)(X_{k-1}^{(1)})(q)\Delta^{-1} \Lambda(q^{-1})
+\sum_{j=1}^{k-1}\left(j-\frac{k-2}{2}\right)\widetilde{L}^{k-j-2}(\widetilde{q}_i)\Delta^{-1}(\widetilde{L}^*)^{j}(\widetilde{r}_i)\\
&&-\sum_{j=1}^{k-1}\widetilde{L}^{k-j-2}(\widetilde{q}_i)\Delta^{-1}(\widetilde{L}^*)^{j}(\widetilde{r}_i)\\
&=&\sum_{j=0}^{k-2}\left(j-\frac{k-2}{2}\right)\widetilde{L}^{k-j-2}(\widetilde{q}_i)\Delta^{-1}(\widetilde{L}^*)^{j}(\widetilde{r}_i)
-\sum_{j=0}^{k-2}\widetilde{L}^{k-j-2}(\widetilde{q}_i)\Delta^{-1}(\widetilde{L}^*)^{j}(\widetilde{r}_i)\\
&&+ T_D(q)(X_{k-1}^{(1)})(q)\Delta^{-1} \Lambda(q^{-1})+\frac{k}{2}T_D(q)(L^{k-2}(q_i)\Delta^{-1}(\widetilde{L}^*)^{0} r_i)(q)\Delta^{-1} \Lambda(q^{-1})\\
&=&\widetilde{X}_{k-1}^{(1)}-(\widetilde{L}^{k-1})_-
+\left(T_D(q)(X_{k-1}^{(1)}+\frac{k}{2}L^{k-1}(q_i))\right)(q)\Delta^{-1} \Lambda(q^{-1}).
\end{eqnarray*}
Here we use the relation of $T_D(q_a)(q_a)=0$.
\end{proof}

\begin{lemma}
\begin{equation}\label{paktd}
    \partial_k^*T_D(q)\cdot
    T^{-1}_D(q)=-\Lambda(q)\Delta{\left(q^{-1}((M_{\Delta}L^k)_++\frac{k}{2}L^{k-1}+X_{k-1}^{(1)})\right)(q)}\Delta^{-1} \Lambda(q^{-1}).
\end{equation}
\end{lemma}
\begin{proof}By (\ref{addsymq}),
\begin{eqnarray*}
 \partial_k^*T_D(q)\cdot T_D(q)^{-1}&=&-T_D(q)\partial_k^*(T_D(q)^{-1})\\
 &=&-\Lambda(q)\Delta q^{-1}\partial_k^{*}(q)\Delta^{-1}\Lambda(q^{-1})-\Lambda(q\partial_k^*(q^{-1}))\\
 &=&-\Lambda(q)\left(\Lambda(q^{-1}\partial_k^{*}(q)) \Delta+\Delta(q^{-1}\partial_k^{*}(q))\right)\Delta^{-1} \Lambda(q^{-1})-\Lambda(q\partial_k^*(q^{-1}))\\
 &=& -\Lambda(q)\Delta(q^{-1}\partial_k^{*}(q))\Delta^{-1} \Lambda(q^{-1})\\
  &=&-\Lambda(q)\Delta{\left(q^{-1}((M_{\Delta}L^k)_++\frac{k}{2}L^{k-1}+X_{k-1}^{(1)})\right)(q)}\Delta^{-1} \Lambda(q^{-1}).
\end{eqnarray*}
\end{proof}

\begin{theorem}\label{addgautd}
The additional symmetry flows (\ref{addsym}) for the cdKP
hierarchy ((\ref{cdkplax}) for $m=1,l=1$) commute with the difference type gauge
transformation $T_D(q)$ preserving the form of cdKP , up to shifting of
(\ref{addsym}) by ordinary time flows, that is,
\begin{equation}\label{cdkpaddgaugetd}
    \partial_k^*\widetilde{L}=\left[-(\widetilde{M}_{\Delta}\widetilde{L}^k)_-+\widetilde{X}_{k-1}^{(1)},\widetilde{L}\right]+\frac{\partial \widetilde{L}}{\partial
    t_{k-1}}.
\end{equation}
\end{theorem}
\begin{proof}Firstly, by (\ref{addsym}),
\begin{eqnarray}
\partial_k^* \widetilde{L}&=&\partial_k^*\left(T_D(q)LT^{-1}_D(q)\right)\nonumber\\
&=& \partial_k^* T_D(q)\cdot
LT_D(q)^{-1}+T_D(q)\partial_k^*L\cdot
T^{-1}_D(q)-T_D(q)LT^{-1}_D(q)\cdot\partial_k^*T_D(q)\cdot T^{-1}_D(q)\nonumber\\
&=& \partial_k^* T_D(q)\cdot
T^{-1}_D(q)\widetilde{L}+T_D(q)[-(M_{\Delta}L^k)_-+X_{k-1}^{(1)},L]T^{-1}_D(q)-\widetilde{L} T^{-1}_D(q)\cdot\partial_k^*T_D(q)\nonumber\\
&=&
\left[T_D(q)\left(-(M_{\Delta}L^k)_-+X_{k-1}^{(1)}\right)T^{-1}_D(q)+\partial_k^*T_D(q)\cdot
T_D(q)^{-1}, \widetilde{L}\right].\label{prooftd1}
\end{eqnarray}

Then with the help of (\ref{transformtd}), (\ref{paktd}), and the
following useful formula (3.9) in \cite{lmh20131},
we have
\begin{eqnarray}
&&T_D(q)\left(-(M_{\Delta}L^k)_-+X_{k-1}^{(1)}\right)T^{-1}_D(q)+\partial_k^*T_D(q)\cdot
T^{-1}_D(q)\nonumber\\
&=&-T_D(q)(M_{\Delta}L^k)_-T^{-1}_D(q)+T_D(q)X_{k-1}^{(1)}T^{-1}_D(q)+\partial_k^*T_D(q)\cdot
T^{-1}_D(q) \nonumber\\
&=&
-T_D(q)(M_{\Delta}L^k)_-T^{-1}_D(q)+\widetilde{X}_{k-1}^{(1)}+\left(T_D(q)(X_{k-1}^{(1)}+\frac{k}{2}L^{k-1}(q_i))\right)(q)\Delta^{-1} \Lambda(q^{-1})
    \nonumber\\
&&-(\widetilde{L}^{k-1})_-
- T_D(q)\left\{(M_{\Delta}L^k)_++\frac{k}{2}L^{k-1}+X_{k-1}^{(1)}\right\}(T^{-1}_D(q))\nonumber\\
&=&-(\widetilde{M}_{\Delta}\widetilde{L}^k)_-+\widetilde{X}_{k-1}^{(1)}-(\widetilde{L}^{k-1})_-,\label{prooftd2}
\end{eqnarray}
where the following relation \cite{lmh20134} is used,
\begin{equation}\label{laxminus}
    (\widetilde{L}^{k-1})_-=\sum_{j=0}^{k-2}\widetilde{L}^{k-j-2}(\widetilde{q})\Delta^{-1}\widetilde{L}^j(\widetilde{r}).
\end{equation}
In the above process,
\begin{eqnarray*}
&&-T_D(q)(M_{\Delta}L^k)_-T^{-1}_D(q)- T_D(q)(M_{\Delta}L^k)_+(T^{-1}_D(q))\\
&=&-T_D(q)(M_{\Delta}L^k)T^{-1}_D(q)+T_D(q)(M_{\Delta}L^k)_+T^{-1}_D(q)- T_D(q)(M_{\Delta}L^k)_+(T^{-1}_D(q))\\
&=&-(\widetilde{M}_{\Delta}\widetilde{L}^k)+(\widetilde{M}_{\Delta}\widetilde{L}^k)_+\\
&=&-(\widetilde{M}_{\Delta}\widetilde{L}^k)_-,
\end{eqnarray*}
 which can be got by means of  the identities (\ref{Td}) of Lemma.\ref{fubu}.

At last, the substituting (\ref{prooftd2}) into (\ref{prooftd1}) gives
rise to
\begin{eqnarray*}
\partial_k^* \widetilde{L}
&=&\left[-(\widetilde{M_{\Delta}}\widetilde{L}^k)_-+\widetilde{X}_{k-1}^{(1)}-(\widetilde{L}^{k-1})_-, \widetilde{L}\right]\nonumber\\
&=&\left[-(\widetilde{M_{\Delta}}\widetilde{L}^k)_-+\widetilde{X}_{k-1}^{(1)},\widetilde{L}\right]
-\left[(\widetilde{L}^{k-1})_-,\widetilde{L}\right] \nonumber\\
&=&\left[-(\widetilde{M_{\Delta}}\widetilde{L}^k)_-+\widetilde{X}_{k-1}^{(1)}, \widetilde{L}\right]+\frac{\partial \widetilde{L}}{\partial
    t_{k-1}}.
\end{eqnarray*}
\end{proof}

\noindent{\bf Remark}: when $m>1$, (\ref{cdkpaddgauge}) will not
hold. In fact, when $m>1$, (\ref{transformedx}) will become into
\begin{eqnarray}
T_I(r_a)X_{k-1}^{(1)}T_I(r_a)^{-1}&=&\widetilde{X}_{k-1}^{(1)}+\sum_{j=0}^{k-2}\widetilde{L}^{k-j-2}(\widetilde{q_a})\Delta^{-1}\widetilde{L}^j(\widetilde{r_a})\nonumber\\
    &&+\Lambda^{-1}(r_a^{-1})\Delta^{-1}\left\{T_I(r_a)^{*-1}(X_{k-1}^{(1)}-\frac{k}{2}L^{k-1})^*(r_a)\right\},
\end{eqnarray}
where $r$ is one of the adjoint eigenfunctions in (\ref{cdkplax}).
We can see that
$\sum_{j=0}^{k-2}\widetilde{L}^{k-j-2}(\widetilde{q_a})\Delta^{-1}\widetilde{L}^j(\widetilde{r_a})$
can not be written as $(\widetilde{L}^{k-1})_-$, since when $m>1$
\cite{orlov1996},
\begin{equation*}
    (\widetilde{L}^{k-1})_-=\sum_{a=0}^{m}\sum_{j=0}^{k-2}\widetilde{L}^{k-j-2}(\widetilde{q_a})\Delta^{-1}\widetilde{L}^j(\widetilde{r_a}).
\end{equation*}
Thus from the proof of (\ref{cdkpaddgauge}), the term of
$\partial_{t_{k-1}}\widetilde{L}$ in (\ref{cdkpaddgauge}) can not be
derived.

The same as the integral type gauge transformation,  the difference type gauge transformation (\ref{cdkpaddgaugetd}) also is not satisfied for $m>1$.

\section{Conclusions and Discussions}
After some technique identities of two types gauge transformations of the cdKP hierarchy,
the interplay of the integral type gauge transformation $T_I$ and the difference type gauge transformation $T_D$ with
the additional symmetry at the instance of  the cdKP
hierarchy are gotten in Theorem \ref{addgau}, Theorem \ref{addgautd} (see
(\ref{cdkpaddgauge}, \ref{cdkpaddgaugetd})),  which preserving the form of the additional symmetry of the cdKP hierarchy, up to shifting of
  the corresponding additional flows by ordinary time flows. But it is shifting for different directions for the integral type gauge transformation and the
  difference type gauge transformation of the cdKP hierarchy.
 It reflects one of the intrinsic features for the cdKP hierarchy. These results
provide a mathematical background from the point of view of integrable systems of
the potential applications in physics for the additional symmetry flows of the
cdKP hierarchy.

{\bf {Acknowledgements:}}
  {\small This work is supported by
the National Natural Science Foundation of China  under Grant Nos.11271210 and 11301526, K. C. Wong Magna Fund in
Ningbo University,  Natural Science Foundation of Ningbo under Grant No. 2014A610029 and Science Fund of Ningbo
University (No.XYL14028). One of the authors (MH) is  supported by Erasmus Mundus Action 2 EXPERTS III
and would like to thank Prof. Antoine Van Proeyen for many helps.}

\end{document}